\documentclass[12pt]{article}
\usepackage[utf8]{inputenc}
\usepackage{amssymb}
\usepackage{amsmath}
\usepackage{amsthm}
\usepackage[T1]{fontenc}
\usepackage{tikz}
\usepackage{mathtools}
\usepackage{wasysym}
\usepackage{verbatim}
\usepackage{graphicx}
\usepackage{color}
\usepackage[margin=1.15in]{geometry}
\usepackage{hyperref}

\newtheorem{Lemma}{Lemma}
\newtheorem{theorem}{Theorem}
\theoremstyle{definition}
\newtheorem*{definition}{Definition}

\newtheorem{Corollary}{Corollary}
\newtheorem*{remark}{Remark}
\begin{document}
\newcommand{\problem}[1]{\newpage\subsection{#1}}
\newcommand{\done}[1]{$\hfill\qed$\newline \textbf{Collaborators:} #1}
\newcommand{\nextProblem}[2]{\done{#1}\problem{#2}}
\newcommand{\mb}[1]{\mathbb{#1}}
\newcommand{\mc}[1]{\mathcal{#1}}
\newcommand{\Image}{\text{Im }}
\newcommand{\Gr}{\text{Gr}}
\newcommand{\Mat}{\text{Mat}}
\newcommand{\Le}{\scalebox{-1}[1]{L}}
\newcommand{\eps}{\varepsilon}
\newcommand{\tld}[1]{\widetilde{#1}}
\newcommand{\ct}[1]{\widehat{#1}}

\title{An Optimal Algorithm for Online Freeze-tag}
\date{\today}
\author{Josh Brunner\\ brunnerj@mit.edu\\ Cambridge, MA \and  Julian Wellman\\ wellman@mit.edu \\ Cambridge, MA}

\maketitle

\begin{abstract}
In the \emph{freeze-tag problem}, one active robot must wake up many frozen robots. The robots are considered as points in a metric space, where active robots move at a constant rate and activate other robots by visiting them. In the (time-dependent) \emph{online} variant of the problem, frozen robots are not revealed until a specified time. Hammar, Nilsson, and Persson have shown that no online algorithm can achieve a competitive ratio better than $7/3$ for online freeze-tag, and asked whether there is any $O(1)$-competitive algorithm. In this paper, we provide a $(1+\sqrt{2})$-competitive algorithm for online time-dependent freeze-tag, and show that no algorithm can achieve a lower competitive ratio on every metric space.
\end{abstract}

\section{Introduction}
In the freeze-tag problem (FTP), there are $n$ robots in a metric space. Each robot is either awake (active) or asleep (frozen), and initially only one is awake. The goal is to get all the active robots to wake up all the asleep robots in the minimum possible time. Whenever an active robot reaches an asleep robot, that robot wakes up and can help wake up additional robots. All robots move at the same constant rate, but only active robots may move. A solution to the problem consists of a route which wakes up all of the robots, and is optimal if it wakes up all the robots in the minimal possible time.

The freeze-tag problem can be interpreted as finding a minimum-depth directed spanning tree on a set of points, where each vertex has out-degree at most 2. The first work was done in this language, e.g. in~\cite{konemann05}, and was motivated by (for example) the IP multicast problem, where a server needs to distribute information to a set of hosts. The freeze-tag problem was first introduced under this name in~\cite{arkin06}, and finding an optimal solution was shown to be NP-hard. Further work has mostly centered around approximation algorithms such as in~\cite{arkin03improved} and~\cite{sztainberg02}. No PTAS for general metric spaces has been found, though much progress has been made for Euclidean metrics in \cite{moezkarimi14} and \cite{yazdi15}, finding a linear-time PTAS in some cases. In~\cite{arkin06} it is shown that $5/3$-approximation is NP-hard for general metrics arising from weighted graphs, so a PTAS does not exists assuming $P\neq NP$.

In this paper, we focus on the online version of this problem, where the asleep robots, or requests, are not known in advance. Each request is released at a certain time, before which the location, time, or existence of the request is not known. The goal is to minimize the time when the last robot is awakened. This variant was named \emph{time-dependent} freeze-tag (TDFT) by Hammar, Nilsson, and Persson~\cite{hammar06}. We feel that this variant models the schoolyard game of freeze-tag more closely, since one doesn't know where or when the next person will get tagged, and also may be more relevant for some applications, where, say, requests for information are unpredictable.

Hammar, Nilsson, and Persson are concerned with the \emph{competitive ratio} achieved by an online algorithm, to model the worst-case performance. They show that no algorithm can achieve a competitive ratio lower than $7/3$, by giving a specific metric where this is not possible~\cite[Theorem 5]{hammar06}. They ask whether there is any online algorithm that achieves a constant competitive ratio. We will slightly improve their bound (through a generalized construction) to show that no competitive ratio lower than $1+\sqrt{2}$ is possible, and give an algorithm which achieves this ratio.

\section{Setup and Results}

Recall that a metric space $M$ is a set equipped with a distance function $d : M\times M \to \mb R$ which is non-negative, symmetric and satisfies the triangle inequality. Metric spaces induce a topology generated by the open balls $B_\eps(x)\coloneqq  \{y\in M : d(x,y) < \eps\}$, for $\eps > 0$ and $x\in M$. Examples of metric spaces include those arising from weighted graphs satisfying the triangle inequality, or, say, Euclidean spaces. For time-dependent freeze-tag, we will exclusively use a special class of metric spaces.
\begin{definition}\label{metrics}
A metric space $(M,d)$ is \emph{strongly connected} if for any two points $x,y \in M$, there exists a continuous function $f: [0,d(x,y)]\to M$ with $f(0) = x$ and $f(d(x,y)) = y$, and for all $z \in [0,d(x,y)]$, we have $d(f(z),x) + d(f(z,y)) = d(x,y)$.
\end{definition}
Intuitively, a metric space is strongly connected if for every pair of points $x,y\in M$ there is a path from $x$ to $y$ of length $d(x,y)$. If a connected metric space $(M,d)$ is not strongly connected, there is a metric space $(M,d')$ which is, where $d'(x,y)$ is the minimal value over all (continuous) paths from $x$ to $y$ of the maximal value of $d(w,x) + d(w,y)$, where $w$ ranges over all points on the path. This modification is natural, because it makes the distance function actually reflect the time required to move between points. 

An instance of the freeze-tag problem consists of a list of the $n$ positions $p_0,\ldots,p_{n-1}$ of the robots in a metric space $M$. The point $p_0$ denotes the starting position of the one active robot, while $p_1,\ldots,p_{n-1}$ are the positions of the frozen robots. We refer to the robot which began at position $p_i$ by $r_i$. Solutions to the freeze-tag problem can be considered as binary trees rooted at $p_0$ which span the positions $p_i$. Then the FTP is equivalent to finding the spanning binary tree rooted at $p_0$ which has the minimal possible weighted depth.

An instance of the online (time-dependent) problem consists of the same points $p_i$, but with associated release times $t_i$. We assume that $0 = t_0 \leq t_1 \leq \cdots \leq t_n$. The offline (normal FTP) problem is the special case where $t_i = 0$ for all $i$. The robots are denoted $r_i = (p_i,t_i)$. Time-dependent freeze tag necessarily takes place in a strongly connected metric space, so that robots can take paths between any two points $x,y$, which take $d(x,y)$ time to complete. We do not lose much by restricting to strongly connected spaces, because as noted above any connected space can be modified to be strongly connected, and it doesn't make sense to play freeze-tag on disconnected spaces.

A solution to an instance of the TDFT problem consists of a path in the metric space which unfreezes each robot, while an algorithm for TDFT gives a strategy which says how to move the active robots in any instance of the problem, possibly depending on the metric. An optimal solution to an instance consists of a \emph{optimal scheduling tree}, which unfreezes each robot no earlier than it is released, and minimizes the time which the last robot is unfrozen. The problem of finding the optimal scheduling tree for a given input is NP-hard \cite{arkin06}, but it is at least computable. We seek to minimize the \emph{competitive ratio} of an algorithm $A$ for time-dependent freeze-tag. For each instance $\sigma$ for TDFT, there is an associated time required for the optimal scheduling tree, denoted $OPT(\sigma)$, and a time which the algorithm's solution takes, $A(\sigma)$. We want to minimize the competitive ratio, $R \coloneqq \max\limits_\sigma \frac{A(\sigma)}{OPT(\sigma)}$.

In~\cite{hammar06}, Hammar et. al. give an example of a metric space where no algorithm can achieve a competitive ratio lower than 7/3 for the online time-dependent freeze-tag problem. They pose the question of whether there is any algorithm which achieves a constant competitive ratio in every metric space. We answer this question affirmatively, and describe an algorithm which we show achieves the best possible competitive ratio.

\begin{theorem}\label{main}
The algorithm described in Section~\ref{alg} is $(1+\sqrt{2})$-competitive for the online TDFT problem on every continuous metric space. Moreover, for every $\eps > 0$, there exists a continuous metric space where no algorithm is $(1+\sqrt{2} - \eps)$-competitive for the online TDFT problem.
\end{theorem}

In Section~\ref{alg} we describe our algorithm and show it is $(1+\sqrt{2})$-competitive. In Section~\ref{lower}, we describe a metric space which is extremely similar to the one presented in~\cite{hammar06}, but give a different analysis that will prove a weaker bound but provide a framework, base-case, and motivation for a more complicated analysis. In Section~\ref{tight} we generalize the construction to give a family of metrics, and show that the metrics in this family give lower bounds on the competitive ratio that can be arbitrarily close to $(1+\sqrt{2})$, completing the proof of Theorem~\ref{main}.

\section{$(1+\sqrt 2)$-Competitive Algorithm}\label{alg}
The key idea of our algorithm is \emph{patience}; we hope to have the robots wait near their starting positions until all of the robots are released, at which point we can copy the optimal scheduling tree. Ideally, we don't move any robots until a time $t$ such that the optimal scheduling tree for the current input sequence would take time at most $t/\sqrt{2}$, at which point we use this scheduling tree to wake up all of the robots, taking a total time of $t(1+\frac{1}{\sqrt{2}})$, achieving the desired competitive ratio. Since we do not know the ultimate number of robots which will be released, we cannot know when we truly need to start waking up robots, and so the algorithm needs to be a little fancier. Let's describe our algorithm more precisely now.

Let $OPT(j)$ denote the minimum depth of a binary tree rooted at $p_0$ which wakes up all of the robots $r_i = (p_i,r_i)$ for $i \leq j$, which are released by time $t_j$, under the condition that robot $r_i$ is not activated until at least time $t_i$. Equivalently, $OPT(j)$ is the time of the last unfreezing in the optimal scheduling tree for the instance truncated at $r_j$. Every time a robot is released, we recompute the value $OPT(j)$. This may be an NP-hard problem, as shown in~\cite{arkin06}, but it is as least computable. 

Our online algorithm always has a schedule in mind for waking up the swarm. At every moment, all robots follow the current schedule. Every time a robot $r_j$ is released, we overwrite the current schedule with the following new schedule, and start over from step 1. \begin{enumerate}
    \item Send every active robot $r_i$ back to its starting position $p_i$.
    \item Wait until time $t = \sqrt{2} \cdot OPT(j)$.
    \item Wake up the swarm in time $OPT(j)$ by following an optimal schedule.
    \item Send every robot $r_i$ back to its starting position $p_i$, and wait there.
\end{enumerate}
This algorithm appears to be $(1+\sqrt{2})$-competitive, since it should complete waking up the swarm at time $\sqrt{2}\cdot OPT(j) + OPT(j) = (1+\sqrt{2})\cdot OPT(j)$. The main thing which remains to be shown is that the algorithm always completes the first step before time $t = \sqrt{2}\cdot OPT(j)$.
\begin{Lemma}\label{nearby}
Under the algorithm described, at any time $T$, each robot $r_i$ is at a distance at most $\frac{T}{1+\sqrt{2}}$ from its starting position $p_i$.
\end{Lemma}
\begin{proof}
We note that it suffices to only consider times during step 3, since at any other time the robots are either at home or moving toward home.

Step 3 started at time at least $\sqrt{2} \cdot OPT(j)$, so if we have been in step 3 for a duration $d$, the current time is at least $T \geq d + \sqrt{2}\cdot OPT(j)$. Step 3 takes a total of $OPT(j)$ time, so we also have $OPT(j) \geq d$. Therefore $T \geq d + \sqrt{2}\cdot d$, and so $d \leq \frac{T}{1+\sqrt{2}}$. 
Since robots move at unit speed, if a robot has been moving for at most $d$ time since the last time it was at its starting position, it must be within $d$ distance of its starting position. Since step 3 always begins with all robots at the starting position, it follows that each robot is always within $\frac{T}{1+\sqrt{2}}$ of its starting position.
\end{proof}

We are now prepared to prove half of our main result.
\begin{theorem}\label{2.414}
The algorithm described is $(1+\sqrt{2})$-competitive under any metric for the time-dependent online freeze-tag problem.
\end{theorem}
\begin{proof} Let $OPT$ be the time that an optimal schedule would need to wake up all of the robots. Note that if the final request is released at a time $t$, then $OPT\geq t$, since it is not possible to satisfy a request before it has been received. Thus, there will be no further requests released after time $OPT$.

Consider the time $\sqrt{2}\cdot OPT$. Since the last time we received a request was at the latest at time $OPT$, we have not modified the schedule since then. Thus, we have had at least $(\sqrt{2}-1)\cdot OPT=\frac{OPT}{1+\sqrt{2}}$ time to complete step 1 of the algorithm. By Lemma~\ref{nearby}, at time $OPT$, each robot is at most $\frac{OPT}{1+\sqrt{2}}$ away from its starting location. Thus, by time $\sqrt{2}\cdot{OPT}$, step 1 will have finished, and all the robots will be at their starting locations. Then step 3 will begin at time $\sqrt{2} \cdot OPT$, and take at most $OPT$ time, for a total time of $(1+\sqrt{2})\cdot OPT$ for when the last robot is awakened.
\end{proof}
\begin{remark}
Our algorithm balances staying close to home with having to wait before executing an optimal schedule. Any algorithm for which Lemma~\ref{nearby} holds for a smaller constant than $\frac{1}{1+\sqrt{2}}$ will necessarily have a larger competitive ratio. If there were an algorithm with a better competitive ratio, it seems that it must be willing to send a robot farther away, while keeping others slightly closer. Our proof in the next two sections can be viewed as giving a formal justification for why it isn't viable to keep most of the robots closer to home.
\end{remark}

One drawback of our algorithm is that computing $OPT(j)$ takes an exponential amount of time, and so it is not particularly efficient in that sense. One could use an approximation algorithm for $OPT(j)$, which would increase the competitive ratio by the approximation factor. However, even $5/3$-approximation is $NP$-hard, and the best known polynomial-time algorithm for general metrics is only a $(\log n)$-approximation~\cite{arkin06}, so more work is required for our algorithm to be executable efficiently.

\section{Example Lower Bound Construction}\label{lower}

In this section, we use a metric which is very similar to the metric described in~\cite{hammar06} which gave a lower bound of 7/3. This metric can give a lower bound on the optimal competitive ratio of $\frac{\sqrt{33}-1}{2}\approx 2.37228\ldots$ by optimizing Lemma~\ref{dichotomy}, which would by itself be an improvement on the $7/3$ bound. We will present a simplified analysis which only gives a lower bound of $3\sqrt{2} - 2 \approx 2.24264\ldots$, but which aligns better with the methods in Section~\ref{tight}. The complicated family of constructions there can be viewed as generalizing the metric used here, and the analysis will follow a similar strategy. In fact, the metric in this section and the Lemmas proved will serve as the base case for an induction argument. We also will use Lemma~\ref{GetBound} as a framework for proving stronger bounds.

A lower bound of $2$ is easily achieved by not revealing where a single frozen robot is until a time equal to its distance from $p_0$, adversarially placing it opposite whatever direction $r_0$ had been moving before that time. To improve on this, Hammar et. al.~\cite{hammar06} force two robots to move in one direction, then travel all the way back the other way to unfreeze the final robot. Roughly speaking, the improvement in our approach comes from forcing the robots to move a little bit farther before turning back.

Our metric space $M$ is formed by a weighted graph with 8 vertices. The origin $p_0$ is where the initial active robot starts, and there will be 7 other points $p_1,\ldots,p_7$, which aren't necessarily starting points for robots. We place $p_1,\ldots,p_6$ on edges at distance 1 from the origin, while $p_7$ is at a distance $r \coloneqq 1 +\sqrt{2}$ from the origin. We also add edges of length one connecting points $(p_1,p_2),(p_3,p_4),$ and $(p_5,p_6)$, forming three equilateral triangles with the origin. The metric consists of all points along edges in this graph, with distances given by shortest path between the points in the graph, and is pictured in the Figure 1 below. The construction in~\cite{hammar06} has exactly the same structure, but with different edge lengths.

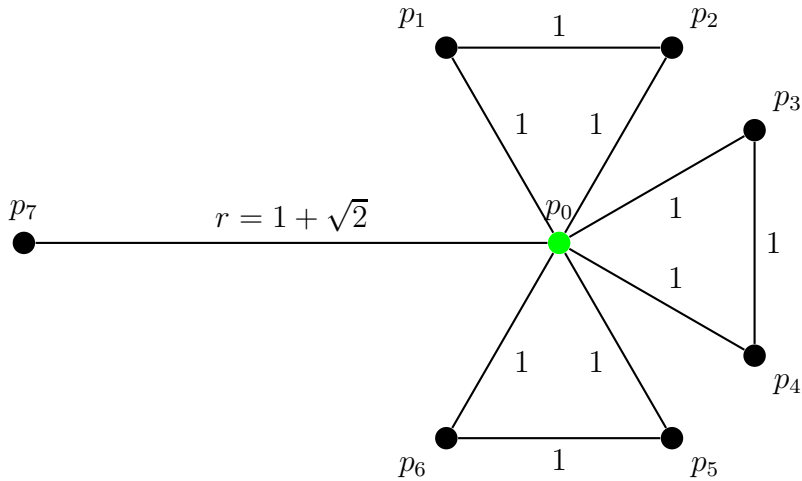
\begin{figure}[h]\label{metric}
\centering
\begin{tikzpicture}[scale=3]
\node [circle,fill=green,inner sep = 3pt, minimum size=5pt,label=above:{$p_0$}] (0) at (0,0) {};
\node [circle,fill=black,inner sep = 3pt, minimum size=5pt,label=above left:{$p_1$}] (1) at (-0.5,0.866) {};
\node [circle,fill=black,inner sep = 3pt, minimum size=5pt,label=above right:{$p_2$}] (2) at (0.5,0.866) {};
\node [circle,fill=black,inner sep = 3pt, minimum size=5pt,label=above right:{$p_3$}] (3) at (0.866,0.5) {};
\node [circle,fill=black,inner sep = 3pt, minimum size=5pt,label=below right:{$p_4$}] (4) at (0.866,-0.5) {};
\node [circle,fill=black,inner sep = 3pt, minimum size=5pt,label=below right:{$p_5$}] (5) at (0.5,-0.866) {};
\node [circle,fill=black,inner sep = 3pt, minimum size=5pt,label=below left:{$p_6$}] (6) at (-0.5,-0.866) {};
\node [circle,fill=black,inner sep = 3pt, minimum
size=5pt,label=above:{$p_7$}] (7) at (-2.372,0) {};

\path [thick]
    (0) edge node[above right] {$1$} (1)
    (0) edge node[above left] {$1$} (2)
    (0) edge node[below right] {$1$} (3)
    (0) edge node[above right] {$1$} (4)
    (0) edge node[below left] {$1$} (5)
    (0) edge node[below right] {$1$} (6)
    (0) edge node[above] {$r = 1 + \sqrt{2}$} (7)
    (1) edge node[above] {$1$} (2)
    (3) edge node[right] {$1$} (4)
    (5) edge node[below] {$1$} (6);
\end{tikzpicture}
\caption{The metric space $M$}
\end{figure}

We can now start describing a TDFT instance for this metric. First fix an algorithm $A$ on the metric $M$. Let $r_0 = (p_0,0)$ and also $r_1 = (p_0,0)$, so that we will have $N\coloneqq 2$ active robots available from the start. Do not release any other robots until time $1$. At time $t=1$, there must a triangle where neither of the robots are. Since the triangles are all the same, without loss of generality we will assume that this is the triangle $p_0p_3p_4$. Then release one robot each at $p_3$ and $p_4$ at time $t=1$. In summary, we define the input $\sigma_A \coloneqq (p_0,0),(p_0,0),(p_3,1),(p_4,1)$. Certainly the time $A(\sigma_A)$ which the algorithm takes to unfreeze both frozen robots is at least 2, since both robots are at least one away from both frozen robots at $t=1$. The following lemma will be key for our analysis.
\begin{Lemma}\label{dichotomy}
For any online algorithm $A$ on the metric $M$ with input $\sigma_A$, we have either $A(\sigma_A) \geq 3\sqrt{2} - 2 =: R$, or there is some time $1+\sqrt{2} \geq t \geq 2$ such that all but at most $N-2\coloneqq 0$ robots are at a distance more than $t(\sqrt{2} - 1)$ from $p_0$, and closer to a frozen robot than $p_0$ is.
\end{Lemma}
\begin{proof}

Suppose at every time $t\in [2,1+\sqrt{2}]$ there is a robot within $t(\sqrt{2} - 1)$ of the origin. The earliest time that either of the initial requests can be satisfied is time 2. Without loss of generality, let robot $r_0$ satisfy the request at $p_3$. At the point when $r_0$ reaches $p_3$, the only way to finish before time $t=3$ is if the request at $p_4$ is satisfied by our other initial robot $r_1$. When $r_1$ reaches $p_4$, it is at a distance $1$ from $p_0$, which is greater than $t(\sqrt{2}-1)$ when $t < 1+\sqrt{2}$. Then the lemma will hold if $r_1$ is closest to $p_0$, so assume some robot goes back towards $p_0$ from $p_3$ (assume it's $r_0$). The scenario right when $r_0$ arrives at $p_3$ at time $t\geq 2$ is depicted in Figure 2. 

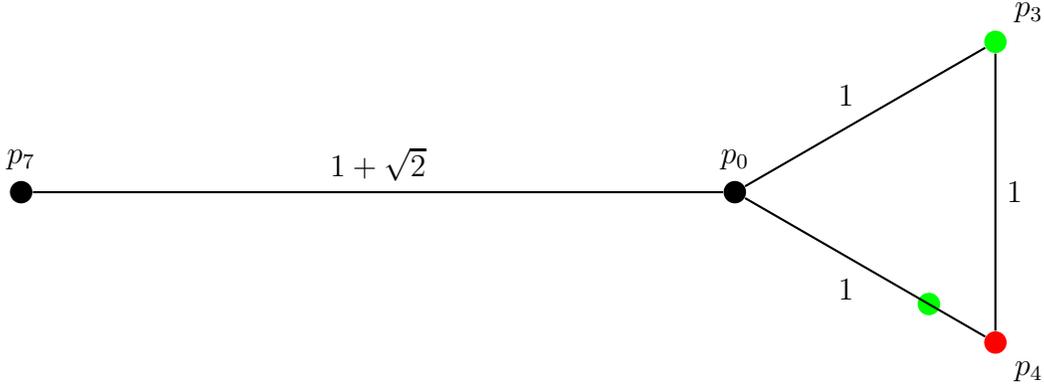
\begin{figure}[h]\label{t2}
\centering
\begin{tikzpicture}[scale=4]
\node [circle,fill=black,inner sep = 3pt, minimum size=5pt,label=above:{$p_0$}] (0) at (0,0) {};

\node [circle,fill=green,inner sep = 3pt, minimum size=5pt,label=above right:{$p_3$}] (3) at (0.866,0.5) {};
\node [circle,fill=red,inner sep = 3pt, minimum size=5pt,label=below right:{$p_4$}] (4) at (0.866,-0.5) {};

\node [circle,fill=green,inner sep = 3pt, minimum size=5pt] (8) at (0.645,-.372) {};

\node [circle,fill=black,inner sep = 3pt, minimum
size=5pt,label=above:{$p_7$}] (7) at (-2.372,0) {};

\path [thick]

    (0) edge node[above left] {$1$} (3)
    (0) edge node[below left] {$1$} (4)

    (0) edge node[above] {$1+\sqrt{2}$} (7)

    (3) edge node[right] {$1$} (4);

\end{tikzpicture}
\caption{Robots at time $t = 2$, where $r_1$ is $2(\sqrt{2} - 1)$ away from $p_0$}
\end{figure}

Let $T$ be the earliest time after reaching $p_3$ that $r_0$ can be within $T(\sqrt{2} - 1)$ of the origin. At that time, $r_0$ will be exactly $1 - T(\sqrt{2} - 1)$ away from $p_3$. Since it left $p_3$ at time at least 2, this means that it is also at most $T-2$ away from $p_3$. Thus, we have that $1 - T(\sqrt{2} - 1) \leq T-2$. Solving this for $T$ gives that $T \geq \frac{3\sqrt{2}}{2}$.

At this time, $r_1$ must still be within $T(1 - \sqrt{2})$ of the origin, since until $T$ it was the only robot available to satisfy our assumption that some robot is within $t(1 -\sqrt{2})$ of the origin. The case where both robots are exactly $T(1 - \sqrt{2})$ away from the origin at time $T$ is depicted in Figure 3.

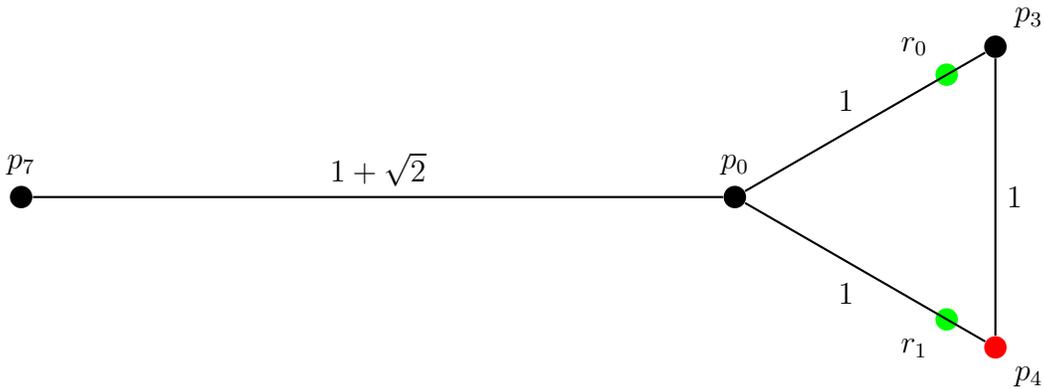
\begin{figure}[h]\label{tT}
\centering
\begin{tikzpicture}[scale=4]
\node [circle,fill=black,inner sep = 3pt, minimum size=5pt,label=above:{$p_0$}] (0) at (0,0) {};

\node [circle,fill=black,inner sep = 3pt, minimum size=5pt,label=above right:{$p_3$}] (3) at (0.866,0.5) {};
\node [circle,fill=red,inner sep = 3pt, minimum size=5pt,label=below right:{$p_4$}] (4) at (0.866,-0.5) {};

\node [circle,fill=green,inner sep = 3pt, minimum size=5pt,label=above left:{$r_0$}] (8) at (0.704,.407) {};
\node [circle,fill=green,inner sep = 3pt, minimum size=5pt,label=below left:{$r_1$}] (8) at (0.704,-.407) {};

\node [circle,fill=black,inner sep = 3pt, minimum
size=5pt,label=above:{$p_7$}] (7) at (-2.372,0) {};

\path [thick]

    (0) edge node[above left] {$1$} (3)
    (0) edge node[below left] {$1$} (4)

    (0) edge node[above] {$1+\sqrt{2}$} (7)

    (3) edge node[right] {$1$} (4);

\end{tikzpicture}
\caption{Robots at time $t = T$, both a distance $T(\sqrt{2} - 1)$ from $p_0$}
\end{figure}

After this time however, now that $r_0$ is sufficiently close to the origin, $r_1$ can immediately walk at full speed to $p_4$. It will still take $r_1$ at least $1 - T(\sqrt{2} - 1)$ additional time to reach $p_4$, for a total time of: \[T + 1 - T(\sqrt{2} - 1) = 1 + T(2-\sqrt{2}) \geq 1 + 3(\sqrt{2} - 1) = 3\sqrt{2} - 2 =:R\]

Thus, as long as there is always a robot within $t(\sqrt{2}-1)$ of the origin, it is not possible to satisfy both requests before time $R\coloneqq 3\sqrt{2} - 2$.
\end{proof}

We'll need to verify some other easy but strange-looking facts.

\begin{Lemma}\label{checkHypotheses}
There exists a schedule for the input $\sigma_A$ on $M$ which unfreezes all robots by time $t=1$, and another schedule where a single robot unfreezes every other robot by time $t = 2$. Moreover, $M$ has $N-1\coloneqq 1$ additional edges of length $1+\sqrt{2}$ connected to $p_0$, on which none of the requests in $\sigma_A$ occur. 
\end{Lemma}
\begin{proof}
The two robots could unfreeze both by going to $p_3$ and $p_4$ right away. One robot could complete $\sigma_A$ by first visiting $p_3$, and then $p_4$, taking $2$ time units. Also, $M$ has the edge connecting $p_0$ and $p_7$ of length $1+\sqrt{2}$ that is not used by $\sigma_A$.
\end{proof}

Now, we can consider the two cases given by Lemma~\ref{dichotomy} to prove a lower bound.

\begin{Lemma}\label{GetBound}
Let $R \leq 1+\sqrt{2}$ be a real number, $N\geq 2$ an integer. Suppose $M$ is a metric such that for any algorithm $A$, there exists an input $\sigma_A$ with $N$ robots at $p_0$ at time $t=0$, such that Lemma~\ref{dichotomy} and Lemma~\ref{checkHypotheses} both hold. Then any online algorithm on $M$ achieves a competitive ratio of at most $R$.
\end{Lemma}
\begin{proof}

Take any algorithm $A$ on $M$ and let $\sigma_A$ be the input given by the hypotheses. By Lemma~\ref{dichotomy}, we know that either $A$ takes at least $R$ time on $\sigma_A$, or there exists some time $1+\sqrt{2} \geq t \geq 2$ when only $N-2$ robots are within $t(\sqrt{2} - 1)$ of the origin. We will consider each of these as separate cases.

\textbf{\underline{Case 1:}} $A$ takes at least $R$ time to complete $\sigma_A$. For $\sigma_A$, we know from Lemma~\ref{checkHypotheses} that the optimal scheduling tree finishes by time 1, so $OPT(\sigma_A) \leq 1$. This gives a competitive ratio of $\frac{A(\sigma_A)}{OPT(\sigma_A)} \geq R$.

\textbf{\underline{Case 2:}} There exists some time $2\leq t \leq 1+\sqrt{2}$ when only $N-2$ robots are closer than $t(\sqrt{2} - 1)$ from $p_0$. Let $p_1,\ldots,p_{n-1}$ be the endpoints other than $p_0$ of the edges of $M$ given Lemma~\ref{checkHypotheses}. Now, we modify $\sigma_A$ to add the additional requests $(p_i\cdot t(\sqrt{2}-1),t)$, which occur at time $t$ along the edge from $p_0$ to $p_i$ located at a distance of $t$ away from $p_0$, for $i = 1,2,\ldots,n-1$. Then the optimal schedule can complete in time $t$. It starts by sending one robot to complete $\sigma_A$ by time $2$ using Lemma~\ref{checkHypotheses}, and the other $N-1$ robots to complete the extra requests at time $t\geq 2$.

For $A$, when the last request is released at time $t$, there are only $N-2$ robots closer than $t(\sqrt{2} - 1)$ to $p_0$. If there are any robots on the edges connecting $p_0$ and $p_i$ for $i = 1,\ldots,n-1$, then they are no closer than $p_0$ is to a frozen robot, so by Lemma~\ref{dichotomy}, they are counted among the $N-2$. Now, $N-2$ robots cannot complete the additional requests in time less than $2t$, which is far too slow. Therefore some robot not among the $N-2$ must unfreeze one of the new robots. Combining the hypotheses of Lemmas~\ref{dichotomy} and~\ref{checkHypotheses}, the shortest route this robot can take goes through $p_0$. This robot (and therefore $A$) must spend a total time of at least $t + t(\sqrt{2} - 1) + t = t(1+\sqrt{2})$ total time to satisfy that request, giving a competitive ratio of at least $\frac{t(1+\sqrt{2})}{t} = 1+\sqrt{2}$.

Since $R \leq 1+\sqrt{2}$, the competitive ratio for any algorithm $A$ is not less than $R$.
\end{proof}

If we apply Lemma~\ref{GetBound} to our particular metric $M$ and $\sigma_A$, it proves that no competitive ratio better than $R\coloneqq 3\sqrt{2} - 2$ is possible. As said before, we could optimize the analysis in this case to show a better bound, but since we will prove a tight bound in the next section, we won't bother. Since we have phrased Lemma~\ref{GetBound} in such a general manner, it suffices to construct metrics $M$ where Lemma~\ref{checkHypotheses} holds and Lemma~\ref{dichotomy} can be proven for a smaller values of $R$.

\section{Tight Lower Bound}\label{tight}
Let $k$ be a non-negative integer. We will define a metric $M_k$ with parameters $N_k$, a natural number, and $T_k$, a rooted tree on $N_k$ vertices that has depth $k$, both of which are a function of $k$.

Let's construct a weighted graph which will form our metric $M_k$. Let $p_1,p_2,\ldots,p_{N_k-1}$ be vertices of degree one connected to a vertex $p_0$ by edges of length $1 +\sqrt{2}$. Then make $N_k+1$ copies of the tree $T_k$ (to be described), rooted at vertices $p_{N_k},p_{N+1},\ldots,p_{2N_k}$. Finally, connect all vertices in these trees to $p_0$ by edges of length 1.

We will describe the tree $T_k$ recursively. We will define $N_k$ to always be the number of vertices in the tree $T_k$, and so also achieve a recursive description of $N_k$. The tree $T_0$ is a single point, and $T_1$ consists of two vertices connected by an edge of length one. Then for $k>1$, let $T_{k+1}$ be rooted at a vertex $v_0$, with descendants $v_1,\ldots,v_{N_k}$. All of the edges connecting $v_0$ and $v_i$ have length $1/2$ For all $i$, let $v_i$ be the root of a copy of $T_k$ with all edge lengths halved. This completes the description of $T_k$ and $N_k$, and so also $M_k$. A sketch of the tree $T_k$ can be seen in Figure~4. Observe that $N_{k+1} = N_k^2 + 1$, that $T_k$ has $k+1$ layers $0,1,\ldots,k$ (where layer $0$ is just the root), but that any path from the root to a leaf has length exactly $1$, for $k\geq 1$. Also, note that $M_1$ is exactly the metric used in the Section~\ref{lower}.

\begin{figure}[h]\label{tT}
\centering
\begin{tikzpicture}[scale=1.65]
\node [circle,fill=black,inner sep = 3pt, minimum size=5pt,label=above:{$v_0$}] (0) at (0,0) {};

\node [circle,fill=black,inner sep = 3pt, minimum size=5pt,label=above:{$v_1$}] (1) at (-4,-3) {};

\node [circle,fill=black,inner sep = 3pt, minimum size=5pt,label=above:{$v_2$}] (2) at (-1.9,-3) {};

\node [circle,fill=black,inner sep = 3pt, minimum size=5pt,label=above left:{$v_3$}] (3) at (-0.3,-3) {};

\node [circle,fill=black,inner sep = 3pt, minimum size=5pt,label=above:{$v_{N_{k-1}}$}] (k) at (4,-3) {};

\node [circle,fill=black,inner sep = 1pt, minimum size = 2pt, label=below:{$N_{k-1}$}] (a) at (0.85,-2) {};

\node [circle,fill=black,inner sep = 1pt, minimum size = 2pt,] (b) at (1.2,-2) {};

\node [circle,fill=black,inner sep = 1pt, minimum size = 2pt,] (c) at (0.5,-2) {};

\node [circle,label=below:{$T_{k-1}$}] at (4,-3.3) {};

\node [circle,label=below:{$T_{k-1}$}] at (-4,-3.3) {};

\node [circle,label=below:{$T_{k-1}$}] at (-1.9,-3.3) {};

\node [circle,label=below:{$T_{k-1}$}] at (-0.3,-3.3) {};


\path [thick]

    (0) edge node[above left] {$1/2$} (1)
    (0) edge node[above left] {$1/2$} (2)
    (0) edge node[above left] {$1/2$} (3)
    (0) edge node[above right] {$1/2$} (k);

\draw [thick] (-4,-3) -- (-4.5,-4) -- (-3.5,-4) -- cycle;

\draw [thick] (4,-3) -- (4.5,-4) -- (3.5,-4) -- cycle;

\draw [thick] (-1.9,-3) -- (-1.4,-4) -- (-2.4,-4) -- cycle;

\draw [thick] (-0.3,-3) -- (-0.8,-4) -- (0.2,-4) -- cycle;

\end{tikzpicture}
\caption{The tree $T_{k}$, in terms of $T_{k-1}$}
\end{figure}
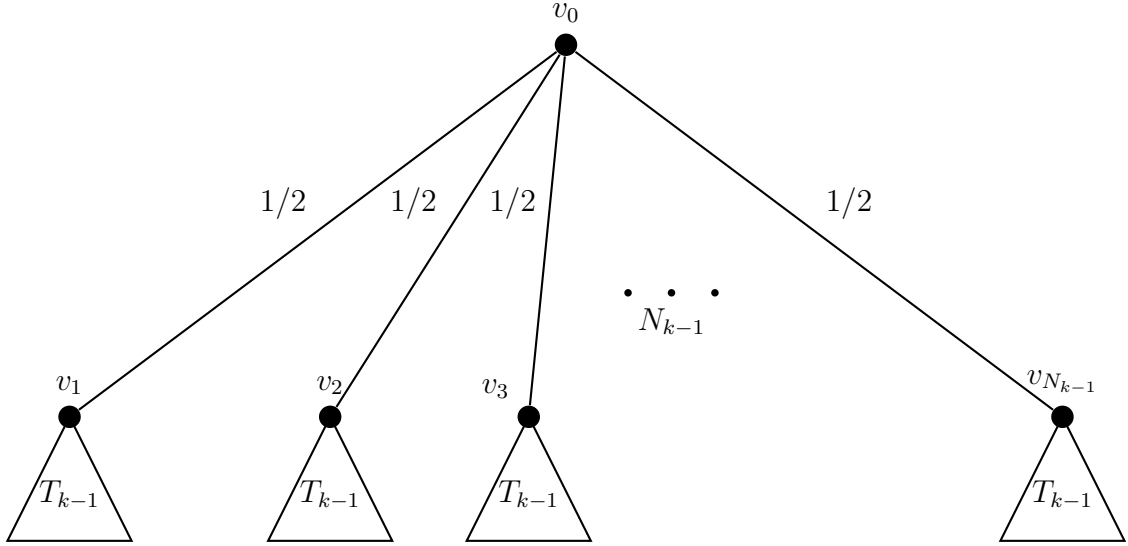

Let's define the input $\sigma_A$, for $A$ an arbitrary online algorithm for the metric $M_k$. Suppose there are $N_k$ starting robots active at $p_0$ at time 0, and no frozen robots. Observe the positions of the robots at time $t=1$. Then since there are $N_k+1$ copies of $T_k$, at least one tree will have no robots along any of the edges to any of the vertices in the tree. Without loss of generality assume this is the tree rooted at $p_{N_k}$. Then let $\sigma_A$ be the input which starts $N_k$ robots active at $p_0$ at time $0$, and releases robots at time $t=1$ at each node of the tree rooted at $p_{N_k}$. Namely, if a node of $T_k$ has down-degree $d$, then release $\max(d-1,1)$ frozen robots at that node. The only exception will be the root of $T_k$, which has down-degree $N_k$, but $\sigma_A$ releases $N_k$ frozen robots at the root instead of $N_k-1$ (this extra robot will make the analysis slightly more clean).
The idea behind this construction of $\sigma_A$ is that it provides just enough robots such that if the root were unfrozen, the robots could cascade through the tree unfreezing everything in 1 unit time.

\begin{Lemma}\label{robotCount}
The number of frozen robots released by $\sigma_A$ in layers $0,\ldots,i-1$ of the tree is equal to the number of nodes in layer $i$.
\end{Lemma}
\begin{proof}
We use induction on $i$. For the base case $i=1$, the number of frozen robots in layer $0$ is just the number of frozen robots at the root, which by construction is $N_k$, which is also the number of nodes in layer 1 of $T_k$.

Suppose the lemma holds for some $i\geq 1$. In layer $i$, each node has down degree $d$ and $d-1$ robots located at it (since $i\geq 1$). Thus, the total number of nodes in layer $i$ layer plus the number of robots in layer $i$ is equal to the number of nodes in layer $i+1$. By our inductive hypothesis, then, the total number of robots in layers $0,\ldots,i$ is equal to the number of nodes in layer $i+1$.
\end{proof}
We aim to prove versions of Lemmas~\ref{dichotomy} and~\ref{checkHypotheses} for $M_k$, for some real number $R_k \in [2,1+\sqrt{2}]$ depending on $k$. The previous section provides the base case $k=1$, where $R_1 = 3\sqrt{2} - 2$. Also note that everything holds for $k=0$, when the graph consists of three spokes, and $R_0=2$. In general, we will define $R_k \coloneqq 1+\sqrt{2} - (\sqrt{2} - 1)^{k+1}$. One can check that this matches $R_0$ and $R_1$. This formula has the important properties that $R_{k+1} - R_{k} < 2^{-(k+1)}$ and that $\lim\limits_{k\to \infty} R_k = 1+\sqrt{2}$, which are both clear. This definition isn't just arbitrary though; it arises as the amount of time it takes to bounce back and forth between $T_k$ and the ball of radius $t(\sqrt{2} - 1)$ around $p_0$.

\begin{Lemma}\label{ARRRRR}
Suppose that at all times $t \in [R_k,R_{k+1})$ there are at least $N_k-1$ robots within a distance $t(\sqrt{2} - 1)$ of $p_0$. Then there exist $N_k-1$ robots that are unfrozen at time $R_k$ which do not unfreeze any robots at any time $t \in [R_k,R_{k+1}]$.
\end{Lemma}
\begin{proof}
A similar argument to what follows appeared within the proof of Lemma~\ref{dichotomy}.

At time $R_k$, there exist at least $N_k-1$ robots that are within $R_k(\sqrt{2}-1)$ of $p_0$. If these robots stay within $t(\sqrt{2} - 1)$ of $p_0$, then they will never reach a frozen robot in this time interval, since all frozen robots are at least one away from $p_0$ and $R_{k+1} < 1+\sqrt{2} = 1/(\sqrt{2}-1)$. Then at some time $t$ some robot that either was frozen at time $R_k$ or unfroze another robot must enter the ball of radius $t(\sqrt{2} -1)$, otherwise there will always be these $N_k-1$ robots that never unfreeze other robots. Now, since this robot entering the ball must have come from a point at a distance 1 from $p_0$, the minimal time $x>0$ after $R_k$ required such that $1 - x \leq (x+R_k)(\sqrt{2}-1)$ is $x = \frac{1-R_k(\sqrt{2} - 1)}{\sqrt 2}$. Then, a robot exiting the ball has a distance $x$ remaining to reach a frozen robots, which therefore cannot occur until time at least $R_k + 2x$.

Now, it can be checked that in fact, $R_{k+1} = R_k + \sqrt{2}(1-R_k(\sqrt{2}-1)) = R_k + 2x$. Therefore it is impossible to avoid having $N_k - 1$ robots never unfreeze a robot in this time interval.
\end{proof}
\begin{Lemma}\label{check3}
Lemma~\ref{checkHypotheses} holds for the metric $M_k$ with input $\sigma_A$ and the integer $N_k$.
\end{Lemma}
\begin{proof}

A schedule can unfreeze all the robots released in $\sigma_A$ by time $t=1$ by having each of our starting $N_k$ robots go directly from $p_0$ to a different vertex of the tree rooted at $p_{N_k}$. At time $1$, they will all arrive and wake up all of the robots.

We can also have a single robot unfreeze all of the robots by time $t=2$. First, our one robot moves to $p_{N_k}$ by time $t=1$. Then, each of the newly unfrozen robots moves down the tree to the root of some copy of $T_{k-1}$. There were $N_k$ robots frozen at $p_{N_k}$, so all of the $N_k$ roots can be reached by time $t=3/2$. By induction, $T_{k-1}$ can be traversed by a single robot in one time unit (the base case $k=1$ is clear), but these copies of $T_{k-1}$ have edges of half the length, so all of these robots can be traversed in half a time unit. Therefore the entire tree can be visited by time $t=2$. 

Finally, the edge from $p_0$ to the vertices $p_1, p_2, \ldots, p_{N_k-1}$ give us $N_k-1$ additional edges of length $1+\sqrt 2$ on which no requests from $\sigma_A$ occur.
\end{proof}

It now remains to prove a sufficiently strong version of Lemma~\ref{dichotomy}, so that we can use Lemma~\ref{GetBound}. Call an unfrozen robot \emph{free} at time $t$ if it is more than $t(\sqrt{2} - 1)$ away from $p_0$. In the context of Lemma~\ref{dichotomy}, we are looking for a time $t$ when all but possibly $N_k-2$ robots are free. It will be helpful to use this to constrain how many robots can be free by a certain time. 

\begin{Lemma}\label{free}
Suppose $A$ is an algorithm such that at any time under the input $\sigma_A$, there are at least $N_k-1$ robots which are not free. Then $A$ cannot unfreeze as many robots as there are nodes in layers $0,\ldots,i$ before time $R_i$.
\end{Lemma}
\begin{proof}
We use induction on $i$. The base case, $i=0$, simply says that $A$ cannot unfreeze any robots before time $2$. This is true because $\sigma_A$ chooses to put all of the frozen robots on a tree that no active robot is near. The nearest a robot could be at $t=1$ is $p_0$, which is distance 1 away from any node of the tree.

Let $x_i$ be the number of nodes in layers $0,\ldots,i$. Suppose that our lemma holds for $i-1$. Then just before time $R_{i-1}$, there are at most $x_{i-1}$ free robots, since of our initial $N_k$ robots at least $N_k-1$ of them must be near $p_0$ and not be free.

By Lemma~\ref{ARRRRR}, there must be $N_k-1$ robots which were unfrozen before $R_{i-1}$ that never unfreeze any robots between the times $R_{i-1}$ and $R_i$. In particular, without loss of generality we will assume that these are the $N_k-1$ robots that are required to stay near $p_0$, so the only robots we have available during this interval are the $x_{i-1}$ free robots.

Since $R_i - R_{i-1} < 2^{-i}$, and every edge coming out of a node in one of layers $0,\ldots,i-1$ has length at least $2^{-i}$, it is not possible for any free robot to fully traverse one of these edges between times $R_i$ and $R_{i-1}$. Thus, we can assume these edges don't exist for bounding the number of robots that can be woken up between times $R_i$ and $R_{i-1}$. Ignoring these edges, we have a bunch of subtrees in the bottom layers, each of which have a total of $N_{k-i}$ robots on them, and $x_{i-1}$ single nodes in layers $0,\ldots,i-1$ which each have at least $N_{k-i}$ robots frozen on them.

Each of our $x_{i-1}$ robots can wake up at most all of the robots in either one subtree or one single node. Since every node in layers $0,\ldots,i-1$ has at least $N_{k-i}$ robots at it, and each of the subtrees has at most $N_{k-i}$ robots in it, the most robots that can be woken up is maximized by sending a free robot to each of the nodes in layers $0,\ldots,i-1$.
By Lemma~\ref{robotCount}, the total number of robots in these layers equal to the number of nodes in layer $i$. Thus, we now have a total of $x_i$ free robots, so we cannot have more than $x_i$ free robots before time $R_i$.

If we have $x_i$ free robots, then since we started with 1 free robot, we have unfrozen less than $x_i$ robots by this point.
\end{proof}


\begin{Corollary}\label{check2}
Lemma~\ref{dichotomy} holds on the metric $M_k$ with input $\sigma_A$, the integer $N_k$, and $R_k\coloneqq 1+\sqrt{2} - (\sqrt{2} - 1)^{k+1}$.
\end{Corollary}
\begin{proof}
Suppose that at every time $t \in [2,1+\sqrt{2}]$ there exists at least $N_k-1$ robots at a distance more than $t(\sqrt{2}-1)$ from $p_0$, otherwise the result is immediate. Then we can apply Lemma~\ref{free} with $i=k$ to get that the number of robots unfrozen must not be more than there are nodes in $T_k$ (which is $N_k$) before time $R_k$. Now, the number of frozen robots released in $\sigma_A$ is more than the number of nodes in $T_k$, since each node has at least one frozen robot but most have more. Therefore $A(\sigma) \geq R_k$.
\end{proof}

\begin{proof}[Proof of Theorem~\ref{main}]
\ \newline
First, Theorem~\ref{2.414} says that our algorithm is $(1+\sqrt{2})$-competitive.

Due to Corollary~\ref{check2} and Lemma~\ref{check3}, we can use Lemma~\ref{GetBound} on $M_k$, with input $\sigma_A$, integer $N_k$, and $R_k\coloneqq 1+\sqrt{2} - (\sqrt 2 - 1)^{k+1}$. Then any online algorithm on $M_k$ achieves a competitive ratio of at most $R_k$. Fix $\eps > 0$. Since $\lim\limits_{k\to \infty} R_k = 1+\sqrt{2}$, choose $k$ such that $1+\sqrt{2} - R_k < \eps$. Therefore no algorithm is $(1+\sqrt{2} -\eps)$-competitive on $M_k$.
\end{proof}

Our analysis required trees which have size exponential in $1/\eps$, since $N_k$ is doubly exponential in $k$ and $1+\sqrt{2} - R_k$ is singly exponential in $k$. Could there be an algorithm that is on the order of $(1 + \sqrt{2} - O(\frac{1}{\log n}))$-competitive, for metrics coming from weighted graphs on at most $n$ vertices?

\newpage
\section{Acknowledgements}

The generalization of the metric in Section~\ref{lower} to the construction of Section~\ref{tight} was inspired by a discussion with Yuan Yao. This paper began as a final project in an advanced algorithms course taught by David Karger and Aleksander Madry. We thank them for teaching the course, and for their useful feedback on a draft of this paper. Finally, we thank three peer editors, Leo Castro, Roberto Ortiz, and Tianyi Zeng for their helpful comments on an early draft.

\bibliographystyle{plain}
\bibliography{references}

\end{document}